\documentclass[numbook,runningheads]{article}

\usepackage{graphicx}
\graphicspath{ {Plots/} }
\usepackage{amsfonts}
\usepackage{mathtools}
\usepackage{amsmath}
\usepackage{amssymb}
\usepackage{amstext}
\usepackage{multirow}
\usepackage{amsthm}

\theoremstyle{plain}
\newtheorem{theorem}{Theorem}[section]
\newtheorem{lemma}[theorem]{Lemma}
\newtheorem{corollary}[theorem]{Corollary}
\newtheorem{proposition}[theorem]{Proposition}

\theoremstyle{definition}
\newtheorem{definition}{Definition}[section]

\theoremstyle{remark}
\newtheorem*{remark}{Remark}

\usepackage{amsmath,amssymb,mathtools}
\usepackage{enumitem}
\usepackage[square,sort,comma,numbers]{natbib}

\newcommand{\ud}{\mathrm{d}}
\newcommand{\I}{\mathrm{i}}

\DeclareMathOperator*{\diag}{diag}

\allowdisplaybreaks

\begin{document}

\title{Polynomial Approximation of Discounted Moments
}


\author{Chenyu Zhao         \and
Misha van Beek \and \phantom{bla}
Peter Spreij \and
Makhtar Ba
}



\date{October 29, 2021}

\maketitle

\begin{abstract}
We introduce an approximation strategy for the discounted moments of a stochastic process that can, for a large class of problems, approximate the true moments. These moments appear in pricing formulas of financial products such as bonds and credit derivatives. The approximation relies on high-order power series expansion of the infinitesimal generator, and draws parallels with the theory of polynomial processes. We demonstrate applications to bond pricing and credit derivatives. In the special  cases that allow for an analytical solution the approximation error decreases to around 10 to 100 times machine precision for higher orders. When no analytical solution exists we tie out the approximation with Monte Carlo simulations.
\end{abstract}

\section{Introduction}
\label{intro}

For pricing and hedging applications, the interest is often in calculating the expected value of a discounted function of a stochastic process,
\begin{align*}
    \mathbb{E}\left[e^{-\int_0^t{r(X_s)\ud s}}f(X_t)\middle|X_0=x\right],
\end{align*}
where $f$ describes the contingent claim and $r$ is the risk-free rate. Sometimes another rate may be used for discounting, such as a hazard rate.

In several cases, this expectation has enough structure to allow for analytical or semi-analytical solutions. For example, if the process $X_t$ is an affine process and $r$ is an affine function, then the Fourier transform of $f$ can be used to compute the expectation up to an integral and the solution to a system of Riccati equations \citep{duffie2003affine}. Also, if $X_t$ is a polynomial process as defined by \citet{cuchiero2012polynomial} and $f$ is a polynomial function, and there is no discounting, then a simple analytical expression exists.

This paper introduces an approximation formula that may work in situations where no analytical expression can be found. The functional form of the approximation of order $k$ is 
\begin{align}\label{eq:approx}
    \mathbb{E}\left[e^{-\int_0^t{r(X_s)\ud s}}\langle\bar{f}^k,b^k(X_t)\rangle\middle|X_0=x\right]
    \approx\langle e^{tA_k}\bar{f}^k,b^k(x)\rangle,
\end{align}
where $b^k(x)$ is a vector of certain basis functions (for now we take it $(1,x,x^2,\ldots,x^{k-1})^\top$ for a univariate process on $\mathbb{R}$, but multivariate cases will be considered), and $\bar{f}^k=(f_0,\ldots,f_{k-1})$ is a vector of length $k$ such that the inner product $\langle\bar{f}^k,b^k(X_t)\rangle=\sum_{i=0}^{k-1}f_ib_i(x)$ represents a (polynomial) expression of the contingent claim. The matrix $A_k$ can be derived from the infinitesimal generator of the process and the function $r$. Naturally when we are interested in the $i$-th discounted moment we can choose a basis vector $\bar{f}^k=e_{i}$, $i=0,\ldots,k-1$. Here $e_{i}$ is the vector of length $k$ that has 1 as the entry at the $i$-th position, all other entries being zero. Note that the numbering starts with $i=0$, which corresponds to the monomials $x^i$, also starting with $i=0$. Other choices for the basis functions are equally well conceivable and we will return later to this.

We investigate two primary applications of this approximation. The first is in the calculation of bond prices in short rate models. As the order increases the approximation approaches machine precision, or falls within the margins of a Monte Carlo price when the true bond price has no closed-form expression. This is shown for Cox-Ingersoll-Ross (CIR) \citep{cox2005theory} and Black-Karasinksi \citep{black1991bond} bond prices. Figure \ref{fig:rates1d} previews several orders of magnitude in performance gain over existing numerical techniques. This comparison was made on a simple CIR bond price to illustrate the convergence to the known analytical solution.

\begin{figure}
    \centering
    \includegraphics[width=8cm]{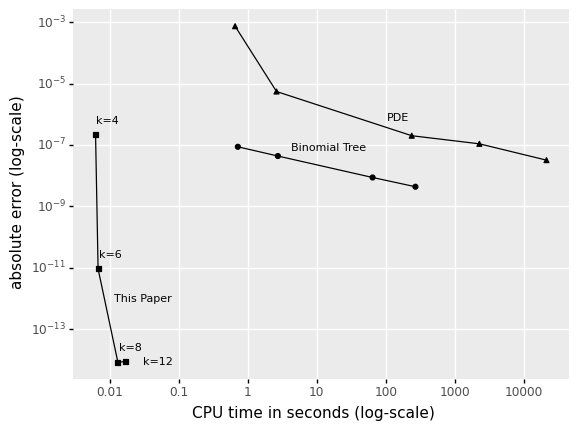}
    \caption{Comparison of polynomial approximation against standard numerical methods for a CIR bond price.}
    \label{fig:rates1d}
\end{figure}

The second application is the calculation of spreads in a generalized Markovian model of credit rating migrations that we develop in an accompanying paper \citep{ba2020integrated}. This model follows the setup of \citet{lando1998cox} and assumes that companies migrate within a set of $m$ ratings, e.g.~$\{\mathrm{AAA},\mathrm{AA},\ldots,\mathrm{CCC},\mathrm{D}\}$, according to a Markov chain $R_t$. The Markov chain has an $m\times m$ generator matrix $Q(Y_t)$ that depends on a state $Y_t$. \citet{jarrow1997markov} derive an analytical solution to spread curves when $Q$ is constant and does not depend on the state $Y_t$. \citet{lando1998cox} and \citet{arvanitis1999building} provide pricing equations for the situation that the generators $Q(Y_t)$ commute (i.e.\ $Q(y_1)Q(y_2)=Q(y_2)Q(y_1)$ for any values $y_1$ and $y_2$ that the process might take), and $Y_t$ follows an affine process. \citet{hurd2007affine} apply these equations to the case that $Q(y)=y_1Q_1+y_2Q_2$ with $Y_t$ bivariate CIR process and $Q_1$ and $Q_2$ are two commuting generator matrices.

The commutativity requirement is highly restrictive, as pointed out by \citet{martin2020credit}. There is strong empirical evidence that upgrades tend to slow down when downgrades speed up, suggesting that the upper and lower triangle of $Q(y)$ are driven separately by two negatively correlated processes. Upper and lower triangular matrices do not commute. We can use our approximation to relax the commutativity requirement as well as the CIR requirement. This relaxation allows us to cover several more stylized facts about credit migrations and spreads, as argued in our accompanying paper \citep{ba2020integrated}. In the exceptional special  cases where an analytical solution exists, our approximation method approaches machine precision as the order increases.

The approximation we propose is easy to compute. It requires application of the infinitesimal generator $\mathcal{A}$ to the terms in a polynomial basis. In the canonical univariate case this means computing $\mathcal{A}b_0,\mathcal{A}b_1,\mathcal{A}b_2,\ldots$ for certain basis functions $b_i$, and subsequently projecting the results on the same basis. Whereas this computation can be done by hand in many cases, in general this is straightforward for a symbolic software package. This is shown in the web appendix to this paper.

Once the correct form of the matrix $A_k$ is found in this way, the computation of moments is very fast. \citet{al2011computing} offer a very efficient and numerically accurate algorithm for calculation the action of the matrix exponential $e^{tA_k}\bar{f}^k$ for a series of times $t$. Subsequent computing of the approximation for a given state $x$ can also be very efficient. In the univariate case, per Horner's method this takes $k$ additions and $k$ multiplications, where we found that usually an order or $k=20$ is sufficient. This is especially convenient for empirical methods such as maximum likelihood estimation, Kalman filtering or MCMC methods, where we typically distinguish between construction of $e^{tA_k}\bar{f}^k$ that happens once per likelihood evaluation, and subsequent computation of the inner product $\langle e^{tA_k}\bar{f}^k,b^k(x)\rangle$ which is required as often as once per time-step within a single likelihood evaluation.

Apart from the applications that this paper explicitly investigates, we list several other applications. First, consider the generic problem of estimating the parameters of a discretely sampled continuous-time stochastic process. Naturally the availability of moments lends itself to generalized method of moments-based estimation, such as described by \citet{zhou2003ito}. But maximum likelihood estimations can also benefit from moment approximations. There is a one-to-one mapping between moments and cumulants, and given some technical conditions, probability density functions can be accurately approximated by cumulants using Gram-Charlier type A series. Such approximations can be much more efficient than PDE, tree, or simulation-based approaches \citep{ait2002maximum}. Second, Gram-Charlier-based approximations of density functions can also be useful for option pricing \citep{popovic2012easy,tanaka2010applications}. Finally, \citet{cuchiero2012polynomial} outline how variance reduction techniques can benefit from knowledge about the moments. 
 
This paper is organized as follows. Section \ref{sec:notation} sets up the general notation of the paper. Section \ref{sec:theory} derives the theoretical results behind our approximation. Finally, sections \ref{sec:appshotrate} and \ref{sec:appshotcredit} show applications to the aforementioned short rate models and credit derivatives respectively. Appendix \ref{appendix} contains some supporting technical results.

\section{General notation}
\label{sec:notation}

This section sets up general notation for the remainder of this paper. We borrow part of our setup from \citet{cuchiero2012polynomial}. Central in our analysis is the Feller process $X$, i.e. a time-homogeneous continuous-time Markov process, with 
state space denoted by $D\subseteq \mathbb{R}^d$. If the process $X$ is not conservative, we augment the state space with a point $\Delta \notin D$ to get the augmented state space $D_\Delta$. 
This point is usually referred to as the \emph{cemetery state} for killed processes, and is used to incorporate discounting. Any function $f$ on $D$ is extended to $D_\Delta$ by the convention $f(\Delta)=0$.
We further consider the  Feller semigroup given $(S_t)_{t\geq 0}$ (often simply denoted $S_t$) given by
\begin{align}\label{eqn:semigroup}
S_t f(x) \coloneqq \mathbb{E}_x[f(X_t)],
\end{align}
and acting on all Borel-measurable functions $f:D_\Delta\to\mathbb{R}$ for which the expectation $\mathbb{E}_x$ is well defined. Here we used $\mathbb{E}_x$ to denote expectation under the law $\mathbb{P}_x$ that is such that the process starts in $x\in D$, i.e.~$\mathbb{P}_x[X_0=x]=1$. When we need (in applications to follow) that certain multi-powers of the $X_t$ have a finite expectation, then it will be implicitly assumed that such moments exist and are finite.
We denote the associated linear operator that describes the process as $\mathcal{A}$, i.e.\
\begin{align}\label{eqn:generator}
    \mathcal{A}f(x)  \coloneqq \lim_{t\downarrow 0}\frac{S_t f(x) - f(x)}{t},
\end{align}
for all functions $f:D_\Delta\to\mathbb{R}$ for which this limit is well defined. This set is the domain of $\mathcal{A}$, denoted $\mathcal{D}(\mathcal{A})$.

This paper relies on series representations, mostly w.r.t.~some orthonormal basis. To start of, we consider some sequence $b$ of linearly independent functions $b_i:D\to\mathbb{R}$, so $b=(b_i)_{i=0}^\infty$. We then further have sequences of real numbers $b(x)=(b_i(x))_{i=0}^\infty$, for $x\in D$.
We denote by $\mathcal{P}$ the space of functions that can be written as a \emph{formal (power) series} with respect to $b$, i.e.~for all $f\in\mathcal{P}$ there exists a sequence $\bar{f}=(\bar{f}_i)_{i=0}^\infty\subset\mathbb{R}$ such that
\begin{align}\label{eq:formal}
f(x)=\langle\bar{f},b(x)\rangle:=\sum_{i=0}^\infty{\bar{f}_ib_i(x)}.
\end{align}
We will need the sum in \eqref{eq:formal} to be convergent in a suitable norm. The sequence $\bar{f}$ then denotes an infinite dimensional vector representation of the sequence and $\langle\cdot,\cdot\rangle$ is the inner product notation for the infinite sum. We will formalize this now.

Consider a separable Hilbert space of functions on $D$, with a certain inner product. Typical examples are the $L^2$ space w.r.t.\ an underlying measure. As a specific example we mention the $L^2$ space of (Borel-measurable) functions $f$ on $\mathbb{R}$ satisfying $\int_\mathbb{R} f(x)^2\phi(x)\,\ud x<\infty$, where $\phi$ is the standard normal density. Clearly, this space contains all polynomials. Moreover, the Hermite polynomials form an orthonormal basis for this space, and choosing the $b_i$ in \eqref{eq:formal} as these polynomials, we have that the squared $L^2$-norm $\| f\|^2$  coincides with $\sum_{i=0}^\infty\bar{f}_i^2$. Taking the $b_i$ as the monic polynomials in \eqref{eq:formal}, then in the same $L^2$-space $\| f\|^2$ can be written as $\bar{f}^\top P\bar{f}$, for a certain strictly positive definite infinite dimensional matrix $P$. 
In what follows, we will always assume, unless stated otherwise, that $\mathcal{P}$ is a Hilbert space of functions w.r.t.~an appropriate inner product $\langle\cdot,\cdot\rangle$ and that it admits an orthonormal base $b$ such that any $f\in\mathcal{P}$ can be represented as in \eqref{eq:formal} with a sum that is convergent in $\mathcal{P}$. We denote by $\mathcal{H}$ the Hilbert space (actually an $\ell^2$-space) of vectors $\bar{f}$ associated with $f\in\mathcal{P}$ for which we impose that $\langle \bar{f},\bar{f}\rangle<\infty$. It follows that the $\ell^2$-norm of $\bar{f}$ coincides with the Hilbert space norm of $f$. With a little, but innocuous abuse of notation, we invariably use the same symbol to denote inner products, sums, norms and thus have $\langle \bar{f},\bar{f}\rangle=\|\bar{f}\|^2=\|f\|^2=\langle f,f\rangle$ and $f=\langle \bar{f},b\rangle$.

By $\mathcal{H}_k$ we denote the subspace of $\mathcal{H}$ of sequences $\bar{f}$ with $\bar{f}_i=0$ for all $i\geq k$. We further let $\mathcal{P}_k$ be the space `polynomials' with $k$ terms (the terminology is suggestive) associated with $\mathcal{H}_k$, meaning that $f\in\mathcal{P}_k$ if $f=\sum_{i=0}^{k-1}\bar{f}_ib_i$. An element $\bar{f}$ of $\mathcal{H}_k$ will be also in a natural way identified with a vector $(f_0,\ldots,f_{k-1})^\top\in\mathbb{R}^k$.
In the case of Hermite polynomials of a single variable, the space $\mathcal{P}_k$ consists of all polynomials of degree $k-1$ or less, hence another basis of $\mathcal{P}_k$ consists of all monomials of order $k-1$ or less. When it is convenient to work with polynomials of a maximum degree, we freely switch between elements of the orthonormal base and monomials. The latter case is an example of a choice of certain basis functions, which in general depend on the state space $D$, but will generally follow standard conventions. For higher-dimensional state spaces $D$, we can use  vector powers $x^k$, where $x\in\mathbb{R}^n$ and $k\in\mathbb{N}_0$, using multi-index notation, as
$x^k = x_1^{k_1} x_2^{k_2}\dots x_n^{k_n}$ for $k_i\in\mathbb{N}_0$, $\sum_{i=1}^n k_i = k$.
Then under this notation, for $D=\mathbb{R}^n$ and $D=\mathbb{R}^n_+$, $b(x)=(1,x,x^2,\ldots)^\top$ has the same symbolic representation as the one-dimensional monomial case. Another useful state space, to which we return later, is the set of basis vectors $D=\{e_1,\ldots,e_d\}$ of $\mathbb{R}^d$. In this case, $b(x)=x$ is an adequate basis as other powers (i.e.~$x^k$ in the vector sense and with $k\neq1$) of unit vectors are linearly dependent. 
A general notational convention follows. We write $b^k(x)=(b_0(x),\ldots,b_{k-1}(x))^\top$, also identified with $b^k(x)=(b_0(x),\ldots,b_{k-1}(x),0,\ldots)^\top$ the vector where the first $k$ entries of $b(x)$ are followed by zeros.

The crux of our approximation theory relies on finite-dimensional modifications of mappings on $\mathcal{H}$. To this end, we introduce some notation that involve projections and subspaces. Let $P_k:\mathcal{H}\rightarrow\mathcal{H}_k$, $k\geq 1$, be a sequence of projection operators, i.e.\ idempotent operators with co-domain $\mathcal{H}_k$. Above we have made the special choice where $\mathcal{H}$ is the $\ell^2$-space of sequences $\bar{f}$ (satisfying $\langle\bar{f},\bar{f}\rangle<\infty$), and $\mathcal{H}_k$ with the space of finite vectors $f^k=(f_0,\ldots,f_{k-1})^\top$ also identified with $(f_0,\ldots,f_{k-1},0,\ldots)^\top$. But Definitions~\ref{def:proj} and \ref{def:stable} extend to the case where $\mathcal{H}$ is  an arbitrary Hilbert space, with the $\mathcal{H}_k$ as subspaces of it.
We borrow further technical assumptions from \citet{kulkarni2008projection}.
\begin{definition}\label{def:proj}
Let $B:\mathcal{H}\rightarrow\mathcal{H}$ be a closed linear operator with domain $\mathcal{D}(B)$. A sequence of bounded projection operators $P_k$ is called \emph{well-behaved for $B$} if
\begin{enumerate}
    \item $P_k\bar{f}\to\bar{f}$ (meaning $\|P_k\bar{f}-\bar{f}\|\to 0$) for all $\bar{f}\in\mathcal{H}$ as $k\to\infty$,
    \item for every $\bar{f}\in\mathcal{D}(B)$, $P_k\bar{f}\in\mathcal{D}(B)$, and
    \item $BP_k\bar{f}\to B\bar{f}$ as $k\to\infty$ for all $\bar{f}\in\mathcal{D}(B)$.
\end{enumerate}
\end{definition}
In \citet{kulkarni2008projection} it is shown that the second requirement of Definition~\ref{def:proj} implies that $P_k\bar{f}\in\mathcal{D}(B)$ for any $\bar{f}\in\mathcal{H}$.
 
In this paper we will consider some natural choices of such projections, with the general $\mathcal{H}$ in this definition taken as  our choice  the space $\mathcal{P}$ of functions introduced above and its corresponding $\ell^2$-space $\mathcal{H}$, along with finite dimensional subspaces $\mathcal{P}_k$ and $\mathcal{H}_k$.

The first is the \emph{finite section} projection, i.e.\ we take the Hilbert space of sequences in $\ell^2$ with $P_k\bar{f}=(\bar{f}_0,\ldots,\bar{f}_{k-1},0,\ldots)^\top$. Correspondingly, if we fix a sequence of basis functions $b(x)$ in $\mathcal{P}$ that forms an orthonormal base and we let $P_k$ be the orthogonal projection on $\mathcal{P}_k$ which is the linear span of $b_0,\ldots,b_{k-1}$, then $P_kf\in\mathcal{P}_k$ has the representation $P_k\bar{f}\in\mathcal{H}_k$. Here we deliberately use the same notation $P_k$ for the projections onto $\mathcal{P}_k$ and $\mathcal{H}_k$.

The second is the \emph{Taylor approximation} around a point $x_0$. E.g.\ with $b(x)=(1,x,x^2,\ldots)^\top$ and $f(x)=x^2$ such that $\bar{f}=(0,0,1,0,\ldots)^\top$ we get a finite section projection $P_2\bar{f}=0$, and a Taylor approximation projection around $x_0$ of $P_2\bar{f}=(-x_0^2,2x_0,0,\ldots)^\top$. The latter follows from ignoring the last term in $f(x)=f(x_0)+f'(x_0)(x-x_0)+\tfrac12 f''(x_0)(x-x_0)^2$, which gives the affine approximation $x_0^2+2x_0(x-x_0)$.
Whether these two projections are well-behaved depends on the linear operator $B$ in Definition~\ref{def:proj}, but this is often easy to verify.

More general, we can consider a fixed sequence of basis functions $b(x)$, that forms an orthonormal base, and projections $P_k:\mathcal{H}\to\mathcal{H}_k$ onto finite dimensional subspaces $\mathcal{H}_k$. Note that for the orthogonal projections $P_k:\mathcal{H}\to\mathcal{H}_k$, given by  $\bar{f}^k\coloneqq P_k\bar{f}\in\mathcal{H}_k$, we have the nice property that the operator norm $\|P_k\|=1$ and $\|\bar{f}_k\|\leq \|\bar{f}\|$.
For any linear operator $A:\mathcal{H}\to\mathcal{H}$ define 
\begin{equation}\label{def:ak}
A_k\coloneqq P_kAP_k|_{\mathcal{H}_k}:\mathcal{H}_k\to \mathcal{H}_k.
\end{equation} 
The restriction to $\mathcal{H}_k$ lets us interpret $A_k$ as a $k\times k$ matrix when $\mathcal{H}_k$ has dimension $k$, as in the case that we just considered. Note that the $A_k$ are automatically bounded operators, whereas $A$ is typically only closed in the cases that are of interest for us.

For closed operators, we also define the notion of stability.
\begin{definition}\label{def:stable}
A family $\{B_k\}_{k=1}^\infty$ of closed operators $B_k:\mathcal{H}_k\rightarrow\mathcal{H}_k$ is stable if there exists a $k_0$ such that $B_k$ is invertible for all $k>k_0$ and $\sup_{k>k_0}{\|B_k^{-1}\|}<\infty$.
\end{definition}

Finally some more notational conventions follow. On finite dimensional spaces, we use the notation $O$, $I$ and $e_i$ to represent the zero matrix, the identity matrix and the $i$-th standard basis vector (the $i$-th column of $I$) respectively. The operator $\otimes$ stands for the Kronecker product, and $\mathrm{diag}(a)$ represents the diagonal matrix, with on the diagonal the elements of a vector $a$.

\section{Polynomial moment approximation theory}
\label{sec:theory}

This section contains the heart of this paper, i.e.\ the theoretical result behind the polynomial approximation that we propose.
Consider a Feller process $X$ on a state space $D$ with Feller semigroup of operators $S_t$ and (infinitesimal) generator $\mathcal{A}$. \begin{definition}\label{def:sequential}
A Feller semigroup $S_t$ and the associated Feller process is called \emph{sequential} if for all $f\in\mathcal{P}$ and $t\geq 0$ (i) $S_tf$ is well defined  and (ii) $S_tf\in\mathcal{P}$. So $\mathcal{P}$ is invariant under each $S_t$.
\end{definition}

\begin{remark}
\citet{cuchiero2012polynomial} call a time-homogenous Markov process $X$ polynomial with semigroup $S_t$ if for all $k\ge0$ we have that $S_tf\in\mathcal{P}_k$ for all $f\in\mathcal{P}_k$ and $t\geq0$. This appears close to being a sequential process, but polynomial processes are not automatically sequential.
\end{remark}
\begin{remark}
The semigroup $S_tf$ appears to represent a standard moment in (\ref{eqn:semigroup}), but can represent a discounted moment when $S_t$ is not a conservative semigroup, namely through appropriate specification of the killing rate at which the process jumps to the cemetery state $\Delta$. For more details see Section~\ref{sec:appshotrate} or \citet[Section 11]{duffie2003affine}.
\end{remark}
\begin{remark}
The time-homogeneity of the Feller assumption can potentially be relaxed to piece-wise time-homogeneity. A typical example happens in the context of local models, where up to a time $\tau_1$ the process $X$ evolves according to a certain semigroup and starting from $\tau_1$ according to another semigroup, and then repeatedly changing at times $\tau_k$. These times are usually chosen to correspond to tenors of derivatives. For some practical applications this is useful, but it complicates notation and analysis considerably, and will not be pursued further in the present paper.
\end{remark}
Any sequential semigroup is a family of linear maps $S_t$ from $\mathcal{P}$ to $\mathcal{P}$, and hence with a fixed basis $b$ these induce linear maps $\bar{S}_t$ from the sequence space $\mathcal{H}$ to $\mathcal{H}$ which have an infinite dimensional matrix representation. Let $g(t)\coloneqq S_tf$. Hence, using the vector representations $\bar{g}(t)$ of $g(t)$ and $\bar{f}$ of $f$, we may write
\begin{align*}
\bar{g}(t)=\bar{S}_t\bar{f},
\end{align*}
where the $ji$-element $\bar{S}_{t,ji}$ is defined as $\bar{S}_{t,ji}=\bar{c}_j^{(i)}(t)$ resulting from the representation of $c^{(i)}(t)\coloneqq S_tb_i$.

In an analogous way we consider the derivative in (\ref{eqn:generator}).
Assuming that each $b_i$ belongs to $\mathcal{D}(\mathcal{A})$, we put $c^{(i)}\coloneqq\mathcal{A}b_i$ and then the $c^{(i)}$ belong to $\mathcal{P}$ as well, i.e.\ $\mathcal{A}:\mathcal{D}(\mathcal{A})\to\mathcal{P}$. 
In all examples that follow, this assumption is satisfied. As all $c^{(i)}$ belong to $\mathcal{P}$, we can represent them by their coordinate vectors $\bar{c}^{(i)}\in\mathcal{H}$ with elements denoted $\bar{c}_j^{(i)}$. We then define the infinite-dimensional matrix $A$ representing a map from $\mathcal{H}$ into $\mathcal{H}$ having $ji$-entry $A_{ji}=\bar{c}_j^{(i)}$. We call $A$ the \emph{matrix generator} of the process $X$. In fact any linear map, call it $\mathcal{A}$ again, from $\mathcal{P}$ into itself naturally induces a map $A:\mathcal{H}\to\mathcal{H}$ in a similar way. As any $f\in\mathcal{P}$ can be identified with a sequence $\bar{f}\in\mathcal{H}$, and similarly a function $g\in\mathcal{P}$ can be identified with a sequence $\bar{g}\in\mathcal{H}$, one can define $\bar{g}=A\bar{f}$ if $g=\mathcal{A}f$.

Since a generator $\mathcal{A}$ of a semigroup is a closed operator, so is $A$. To see this, we use the duality between elements of $\mathcal{P}$ and those of $\mathcal{H}$. We use that $\mathcal{A}$ is closed if $\mathcal{D}(\mathcal{A})$ is complete w.r.t.~the graph norm given by $\|f\|_\mathcal{A}^2=\|f\|^2+\|\mathcal{A}f\|^2$ (see \citep{bobrowski2005functional}, Exercise~7.3.3) and, similarly, that $A$ is closed if $\mathcal{D}(A)$ is complete w.r.t.~the graph norm, which is given by $\|\bar{f}\|_A^2=\|\bar{f}\|^2+\|A\bar{f}\|^2$. But, by construction, $\|f\|_\mathcal{A}=\|\bar{f}\|_A$. In the sequel we will freely switch between $f\in\mathcal{P}$ having an orthogonal expansion in terms of a sequence $\bar{f}$, and between $\mathcal{A} f$ and $A\bar{f}$.

For Feller semigroups we have that $\partial_tS_tf=\mathcal{A}S_tf$. Hence for $g(t)=S_tf$, we have $\partial_tg(t)=\mathcal{A}g(t)$, and then  in the corresponding sequence space $\mathcal{H}$ one has $\partial_t\bar{g}(t)=A\bar{g}(t)$.

Parallelling finite dimensional notation, we also write $\bar{g}(t)=e^{tA}\bar{f}$, as is done for polynomial processes in  \citep[Theorem~2.7]{cuchiero2012polynomial}, although in general the matrix $A$ is genuinely  infinite-dimensional (and has infinite norm). 
We will use the finite dimensional matrix $A_k$ to approximate the semigroup. That is, we use, in ordinary finite dimensional notation,
\begin{align}\label{eq:ggk}
    \bar{g}^k(t)&\coloneqq \bar{S}_t^k\bar{f}^k,&\bar{S}_t^k&\coloneqq e^{tA_k},
\end{align} 
to approximate $\bar{g}(t)$. We us the name `polynomial approximation' as a consequence of the polynomial structure of the approximating $g^k(t)$ in $x$, when the $A_k$ are taken as in \eqref{def:ak}. Our main theoretical result is on the convergence of the approximation in \eqref{eq:ggk}.

\begin{theorem}\label{thm:main}
Consider a sequential process $X$ and a sequence of well-behaved projections $(P_k)_{k=1}^\infty$ (see Definition \ref{def:proj}). If  the $A_k$ as in \eqref{def:ak} have the property that there exists a $\lambda_0$ such that for all $\lambda>\lambda_0$\footnote{Do we actually need this $\lambda_0$? If so, where?} the sequence $(\lambda I-A_k)_{k=1}^\infty$ is stable (see Definition \ref{def:stable}), then $\bar{g}^k(t)\to \bar{g}(t)$, in other notation $e^{tA_k}\bar{f}^k\to e^{tA}\bar{f}$, as $k\to\infty$, with convergence in the $\ell^2$-norm.
\end{theorem}
\begin{proof}
Let $t>0$ as the case $t=0$ is trivial. Consider Phragm\'en's representation, see, \citet{neubrander1987relation} of the semigroup $S_t$ and Lemma~\ref{lemma:pd}. For $f\in \mathcal{D}(\mathcal{A})$ one has 
$S_tf=\lim_{\lambda\to\infty}S_t(\lambda,\mathcal{A},f)$, with, see~\eqref{eq:defs}, 
\begin{align*}
S_t(\lambda,\mathcal{A},f)=\lambda\sum_{n=1}^\infty (-1)^{n-1}\frac{1}{(n-1)!}e^{n\lambda t}R(n\lambda,\mathcal{A})f,
\end{align*}
where $R(\lambda,\mathcal{A})\coloneqq(\lambda I-\mathcal{A})^{-1}$ denotes the resolvent. 
Naturally, when switching from $\mathcal{P}$ to $\mathcal{H}$, we can write this in matrix form, with $g(t)=S_tf$ and $\bar{g}(t)=\bar{S}_t\bar{f}=e^{At}\bar{f}$. One has, see also Corollary~\ref{corollary:pd},
\[
\bar{g}(t)=\bar{S}_t\bar{f}=\lim_{\lambda\to\infty}S_t(\lambda,A,\bar{f}),
\]
with
\begin{align*}
S_t(\lambda,A,\bar{f})=\lambda\sum_{n=1}^\infty (-1)^{n-1}\frac{1}{(n-1)!}e^{n\lambda t}R(n\lambda,A)\bar{f}.
\end{align*}
We can apply the same representation to $\bar{g}^k(t)
=e^{tA_k}\bar{f}^k$,
\begin{align*}
\bar{g}^k(t)
=e^{tA_k}\bar{f}^k
=\lim_{\lambda\to\infty}S_t(\lambda,A_k,\bar{f}).
\end{align*}
Next we embed the finite dimensional vector $\bar{g}^k(t)$, as any element of $\mathcal{H}_k$, in $\mathcal{H}$ simply by appending an infinite sequence of zeros and as such we consider $\bar{g}^k(t)$ as an element of $\mathcal{H}$. Likewise, we can also consider $R(n\lambda,A_k)\bar{f}^k$ as an element of $\mathcal{H}$. Hence we can consider convergence of the $\bar{g}^k(t)$ as elements in $\mathcal{H}$.

Fix some $\lambda>\lambda_0$. We first show that  $R(n\lambda,A_k)\bar{f}^k\to R(n\lambda,A)\bar{f}$. To this end, define the sequence of matrices $\{B_{n\lambda,k}\}_{k=1}^\infty$ by $B_{n\lambda,k}\coloneqq R(n\lambda,A_k)^{-1}=n\lambda I-A_k$. 
Note that $ n\lambda I-A_k=[P_k(n\lambda I-A)P_k]|_{\mathcal{H}_k}$.
Furthermore, every $n\lambda I-A$ is a closed operator because $A$ is closed. By our assumptions in the theorem, the projections $P_k$ are well behaved and the sequence $(B_{n\lambda,k})_{k=1}^\infty$ is stable for all $n\lambda\geq\lambda>\lambda_0$. Then the conditions of \citep[Theorem 3.1]{kulkarni2008projection} are satisfied and so, in the terminology of \cite{kulkarni2008projection} the projection method converges, i.e.~(recall that $\bar{f}^k=P_k\bar{f}$) 
\begin{align*}
    \lim_{k\to\infty} R(n\lambda,A_k)\bar{f}^k=R(n\lambda,A)\bar{f},
\end{align*}
for all $n\geq1$ and $\lambda\geq\lambda_0$, where the limit is taken in $\mathcal{H}$. Having established this convergence, we invoke
Lemma~\ref{lemma:sks} that states that then also
\begin{equation}\label{eq:sum1}
\lambda\sum_{n=1}^\infty(-1)^{n-1}\frac{e^{n\lambda t}}{(n-1)!}R(n\lambda,A_k)\bar{f}^k \to \lambda\sum_{n=1}^\infty(-1)^{n-1}\frac{e^{n\lambda t}}{(n-1)!}R(n\lambda,A)\bar{f}.
\end{equation}
Recall that our aim is to show that the $\bar{S}_t^k\bar{f}^k$ (considered as elements of $\mathcal{H}$) converge to $\bar{S}_t\bar{f}$, where $\bar{S}^k_t=e^{A_kt}$. Therefore, let $\varepsilon>0$ and consider
\begin{align}\label{eq:3things}
\|\bar{S}_t^k\bar{f}^k- \bar{S}_t\bar{f}\| 
& \leq \|\bar{S}_t^k\bar{f}^k- S_t(\lambda,A_k,\bar{f}^k)\| \nonumber\\
&\phantom{=}+ \|S_t(\lambda,A_k,\bar{f}^k)- S_t(\lambda,A,\bar{f})\| \nonumber\\
&\phantom{=}+ \|S_t(\lambda,A,\bar{f})- \bar{S}_t\bar{f}\|. 
\end{align}
It follows from the proof of Lemma~\ref{lemma:sks} that $\|\bar{S}_t^k\bar{f}^k\|\leq \|\bar{S}_tP_k\bar{f}\|$ and hence $\|\bar{S}_t^k\|\leq 1$ as the $\|\bar{S}_t\|$ (as they represent expectations) have norm one.
For $\lambda>\lambda_0$ the first and the last term on the right hand side in \eqref{eq:3things}  are by virtue of Lemma~\ref{lemma:pd} then together less than $\frac{2C}{\lambda}\|\bar{f}\|$. Choose then $\lambda$ such that these terms are both smaller than $\varepsilon$, uniformly in $k$.
For the chosen $\lambda$, the middle term in \eqref{eq:3things} can be made smaller than $\varepsilon$ by choosing $k$ larger than some $k_0=k_0(\lambda_0)$ by Lemma~\ref{lemma:sks}. Hence the total expression on the right of \eqref{eq:3things} is less than $3\varepsilon$ for all $k\geq k_0$. This concludes the proof.
\end{proof}

Theorem \ref{thm:main} dictates what steps should be followed to apply the approximation theory outlined in this section to a Feller process $X$ on state space $D$ with generator $\mathcal{A}$.
\begin{enumerate}
    \item Fix an appropriate basis $b(x)$ for the state space $D$ of the process.
    \item Verify that the process is sequential per Definition \ref{def:sequential}.
    \item Derive $A$ column by column based on $c^{(i)}=\mathcal{A}b_i$.
    \item Choose a projection $P_k$ that is well-behaved per Definition \ref{def:proj}.
    \item Verify the stability criterion in Theorem \ref{thm:main} per Definition \ref{def:stable}.
\end{enumerate}

The first step is generally straightforward. The second step is hard and in the examples in Sections \ref{sec:appshotrate} and \ref{sec:appshotcredit} we will implicitly conjecture that the process is sequential. Step three is again straightforward, although the notation can be somewhat involved for processes on a higher-dimensional state space.

In the examples below, we will choose the Taylor approximation projection unless indicated otherwise, since it tends to converge faster than the finite section projection. It is easy to verify that this projection is well-behaved in all examples.

The fifth and last step proves difficult in practical applications. The best we can do in the examples below is conjecture that the stability criterion holds based on graphical arguments. For example, for a CIR bond price we can calculate the spectral norms of $\|(\lambda I-A_k)^{-1}\|$ under the finite section projection. Figure \ref{fig:norm} shows that there appears to be an upper bound for each $\lambda$. If this is true, then stability is satisfied. 
\begin{figure}
    \centering
    \includegraphics[width=\textwidth]{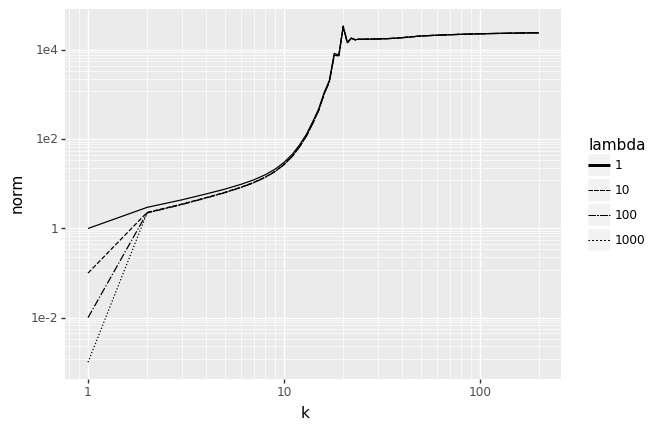}
    \caption{Convergence of the spectral norm for a CIR bond.}
    \label{fig:norm}
\end{figure}

\subsection{Numerical considerations}

In all applications in the sections to follow, we use a standard implementation of the algorithm by \citet{al2011computing} to compute the action of the matrix exponential $e^{tA_k}\bar{f}^k$ directly for a grid of times, since it is much faster and more numerically stable than computing $e^{tA_k}$ separately for several values of $t$ before multiplying by $\bar{f}^k$. In cases where the interest is in multiple moments, $\bar{f}^{k,1},\ldots,\bar{f}^{k,m}$, we use the same algorithm on $e^{tA_k}\bar{F}^k$, with $\bar{F}^k\coloneqq\left(\bar{f}^{k,1},\ldots,\bar{f}^{k,m}\right)$.

In calibration and estimation exercises, the interest may be in the derivative of $\bar{g}(t)$ with respect to the parameters of the process, $\theta$, which are encoded in $A_k$. From the derivative of the exponential map,
\begin{align}
    \frac{\partial}{\partial\theta_i}\bar{g}(t)
    =\left(\frac{\partial}{\partial\theta_i}e^{tA_k}\right)\bar{f}^k
    =\int_0^t{e^{(t-s)A_k}\left(\frac{\partial}{\partial\theta_i}A_k\right)e^{sA_k}\bar{f}^k\ud s}.
\end{align}
The integral can be efficiently approximated by computing the action of the matrix exponential (on $\frac{\partial}{\partial\theta_i}A_k$ and $\bar{f}^k$) for a fine grid of times.


\section{Applications to short rate models}\label{sec:appshotrate}

The next two sections analyze several possible applications of polynomial moment approximation. We start with the bond price approximations of two popular short rate models, before considering more complex credit spread models in the next section.

Many popular short rate models have the following structure: a Feller process $X$ is specified as well as a function $r:D\to\mathbb{R}_+$ such that the short rate at time $t$ is given by $r(X_t)$. In this context, the zero-coupon bond price is given by the expectation
\begin{align}\label{enq:zcbprice}
    P(x,t,T)=\mathbb{E}\left[e^{-\int_t^T{r(X_s)\ud s}}\middle|X_t=x\right]
\end{align}
As shown by \citet{duffie2003affine}, this price is equivalent to the first moment of a modified process with a generator $\mathcal{A}f(x)\coloneqq\mathcal{A}_xf(x)-r(x)f(x)$, with $\mathcal{A}_x$ the infinitesimal generator of the process $X$, and with $r(x)$ the killing rate at which the process jumps to the cemetery state $\Delta$.\footnote{Non-negativity of $r$ on $D$ is required to ensure that the semigroup $S_t$ is Feller, since Feller semigroups are contraction operators. When there is a constant lower bound, i.e.~when the discount rate can be written as $r(x)=\underline{r}+r'(x)$, with $r'$ non-negative on $D$, then we can bring $e^{-\underline{r}(T-t)}$ outside the expectation in (\ref{enq:zcbprice}). However, we have seen that the approximation may also work for negative discounting directly.} With $S_t$ the semigroup corresponding to $\mathcal{A}$, the bond price with respect to basis $b(x)=(1,x,x^2,\ldots)^\top$ is given by
\begin{align}
    P(x,t,T)&=S_{T-t}f(x),&\bar{f}=e_1=(1,0,\ldots)^\top.
\end{align}

Both the Cox-Ingersoll-Ross (CIR) and Black-Karasinksi models both fall in this class. In the CIR case a closed form solution exists, making it an excellent reference case for testing the approximation. In the Black-Karasinksi case we compare with Monte Carlo simulation, since an analytical solution does not exist.

\subsection{The Cox-Ingersoll-Ross bond price}

The CIR one factor short rate model is a popular model to price interest rate derivatives \citep{cox2005theory}. Its state space is the positive real line $D=\mathbb{R}_+$, such that negative rates are avoided. A closed form solution exists for the price of a (zero-coupon) bond to benchmark our approximation.

The CIR short rate dynamics have the following SDE,
\begin{align}
    \ud X_t&= \theta(\mu-X_t) \ud t+\sigma \sqrt{X_t} \ud W_t,&
    r(X_t)&=X_t.
\end{align}
The infinitesimal generator of the modified process is
\begin{equation}
    \mathcal{A}f=\theta(\mu-x)\frac{\partial f}{\partial x}+\frac12 \sigma^2x\frac{\partial^2f}{\partial x^2}-xf,
\end{equation}
where the last term is the adjustment that allows us to compute the bond price as the first moment of the process.

Applying the infinitesimal generator to the base elements $x^i,0\le i \le k-1$ we get
\begin{align}
    \mathcal{A}b_{i+1}(x)=\mathcal{A}x^i
    =-\theta ix^i+\left(\theta\mu i+\tfrac12 \sigma^2i(i-1)\right)x^{i-1}-x^{i+1}.
\end{align}
Only $\mathcal{A}b_k(x)$ has a power of $x$ higher than $x^{k-1}$, the highest power in the polynomial basis $b^k(x)$,
\begin{align*}
    \mathcal{A}b_{k}(x)
    =-\theta (k-1)x^{k-1}+\left(\theta\mu(k-1)+\tfrac12 \sigma^2(k-1)(k-2)\right)x^{k-2}-x^k.
\end{align*}
This higher term $x^k$ is projected onto $b_k(x)$ via Taylor approximation around $x_0=\mu$. The $k-1$-th order approximation can be stored in a $k$-dimensional vector $\bar{p}^k(x_0)$,
\begin{align*}
    \sum_{i=0}^{k-1}{\frac{1}{i!}\left.\partial_x^{(i)}x^k\right|_{x_0}(x-x_0)^i}
    &=\langle\bar{p}^k(x_0),b^k(x)\rangle,
    &\bar{p}_{i+1}^k(x_0)&=-{k\choose i}(-x_0)^{k-i}.\text{\footnotemark}
\end{align*}
\footnotetext{
This follows from the $k$-th order expansion, which is $x^k$, and the binomial theorem,
\begin{align*}
    &\sum_{i=0}^{k-1}{\frac{1}{i!}\left.\partial_x^{(i)}x^k\right|_{x_0}(x-x_0)^i}
    =x^k-\frac{1}{k!}\left.\partial_x^{(k)}x^k\right|_{x_0}(x-x_0)^k\\
    &\qquad=x^k-(x-x_0)^k
    =x^k-\sum_{i=0}^k{{k\choose i}x^i(-x_0)^{k-i}}
    =-\sum_{i=0}^{k-1}{{k\choose i}x^i(-x_0)^{k-i}}.
\end{align*}}
This leads to the following projected matrix generator,
\begin{align*}
A_k&=\begin{bmatrix}
0&\theta\mu&0&0&\cdots&0\\
-1&-\theta&2\theta\mu+\sigma^2&0&\cdots&0\\
0&-1&-2\theta&3\theta\mu+3\sigma^2&\cdots&0\\
\vdots&\ddots&\ddots&\ddots&\ddots&\vdots\\
0&\cdots&0&-1&-(k-2)\theta&(k-2)\theta\mu+\tfrac12(k-2)(k-3)\sigma^2\\
0&\cdots&0&0&-1&-(k-1)\theta
\end{bmatrix}
\\&\phantom{=}+e_k^\top\otimes\bar{p}^k(\mu).
\end{align*}

To show the accuracy of the proposed approximation, we compare it with the analytical solution for eight sets of different but typical parameters.
\begin{figure}
\resizebox{\hsize}{!}{\includegraphics*{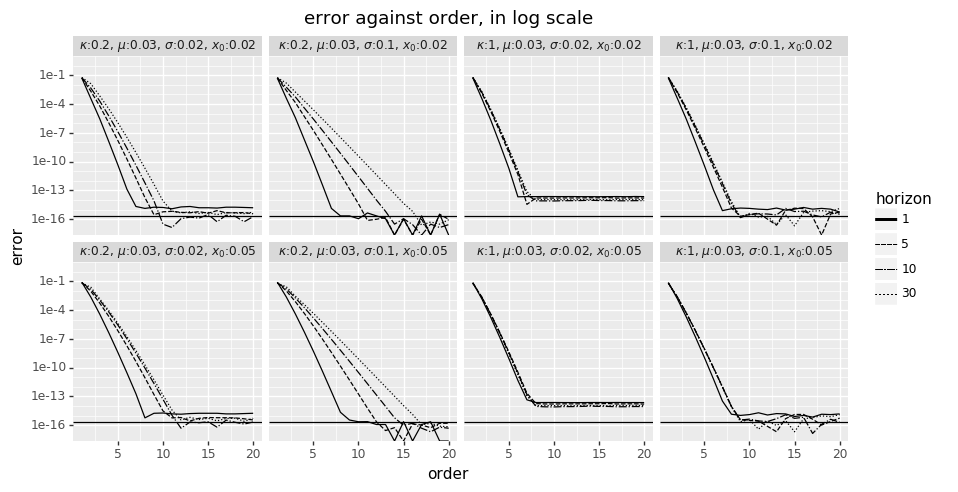}}
\caption{Absolute approximation error of zero-coupon bond price yield. Different lines represent tenors. The horizontal lines indicate machine precision. }
\label{fig:orders_cir}
\end{figure}
Figure \ref{fig:orders_cir} shows the error decreases exponentially as the approximation order $k$ increases. For most sets of parameters, it converges to one or two orders of magnitude above machine precision, and remains stable as the order increases.

\subsection{The Black-Karasinksi bond price}

The Black-Karasinski short rate model is similar in structure to the CIR model, but assumes short rates follow an exponential OU process. No analytical solution is available, hence finding an efficient and accurate approximation has received considerable academic attention. This section shows that our approximation is accurate by means of a Monte Carlo simulation. 

\citet{black1991bond} assume the short rate $r_t$ has dynamics
\begin{align}
 \ud\ln{r_t}=\theta\left(\mu-\ln{r_t}\right) \ud t+\sigma \ud W_t. 
\end{align}
In our setup this translates to
\begin{align}
\ud X_t&= \theta(\mu-X_t) \ud t+\sigma \ud W_t,&
r(X_t)&=e^{X_t}.
\end{align}
The infinitesimal generator of the modified process follows as
\begin{align}
    \mathcal{A}f=\theta(\mu-x)\frac{\partial f}{\partial x}+\frac12 \sigma^2\frac{\partial^2f}{\partial x^2} -e^xf.
\end{align}
Applying this infinitesimal generator to the base elements $x^i$, $0\le i \le k-1$, 
\begin{align}
    \mathcal{A}b_{i+1}(x)=\mathcal{A}x^i=-\theta ix^i+\theta\mu ix^{i-1}+\frac12 \sigma^2i(i-1)x^{i-2}-\sum_{j=0}^\infty{ \frac{x^{j+i}}{j!}}.
\end{align}
There are clearly higher orders than $b_k(x)=x^{k-1}$ present in each expression, as a consequence of $e^x$ in $\mathcal{A}$. As in the CIR case, we use a Taylor projection around the unconditional mean $x_0=\mu$ to derive $A$, but the results are not shown here.




We calculate a set of zero coupon bond yields with different parameters, and use Monte Carlo simulation to benchmark the approximation quality. The unconditional distribution of $X_t$ is Gaussian, with mean $\mu$ and variance $\frac{\sigma^2}{2\theta}$. Thus the steady-state distribution of the short rate $r(X_t)=e^{X_t}$ has mean $\bar{\mu} = \exp(\mu+\frac{\sigma^2}{4\theta})$, and variance $\bar{\sigma}^2 = \big(\exp(\frac{\sigma^2}{2\theta})-1\big)\exp(2\mu+\frac{\sigma^2}{2\theta})$. We fix $\bar{\mu}=0.03$ and vary the other parameters to obtain realistic alternative sets of parameters.
\begin{enumerate}
    \item values of $x=\ln{r_0}$: $\ln{0.01}$, $\ln{0.03}$ and $\ln{0.06}$,
    \item values of $\theta$: 0.02 and 0.1,
    \item values of $\bar{\sigma}$: 6\% and 12\%.
\end{enumerate}
Table \ref{tab:1} outlines the approximation error for maturities 1, 2, 5, 10 and 20 years, and for approximation orders $k=5,10,20$. The approximation error versus Monte Carlo decreases with the order. For order $k=20$, the error is no greater than 1 bps, and lies within the confidence bounds of the Monte Carlo simulation. Figure \ref{fig:orders} gives more graphical insight into the pattern of convergence. For all specifications the errors converge within the confidence intervals of the Monte Carlo simulation. We see a slower convergence for specifications with higher unconditional volatility.


\begin{table}
\caption{Zero coupon bond yield, comparison between Monte Carlo and proposed approximation.}
\label{tab:1}       

\begin{tabular}{cccccccc}
\hline\noalign{\smallskip}
\multicolumn{3}{c}{Parameters}& \multirow{2}{*}{Maturities} & \multirow{2}{*}{MC Yields}&\multicolumn{3}{c}{Errors (bps)}  \\
\noalign{\smallskip}\cline{1-3} \cline{6-8}\noalign{\smallskip}
$\theta$ & $\bar{\sigma}$&$\exp{x}$& & &k=5&k=12&k=20\\
\noalign{\smallskip}\hline\noalign{\smallskip}
\multirow{5}{*}{0.02}&\multirow{5}{*}{0.06}&\multirow{5}{*}{0.01}&1&1.02\%&0&0&0\\
&&&2&  1.04\%&0&0&0\\
&&&5&  1.09\%&0&0&0\\
&&&10& 1.17\%&0&0&0\\
&&&20& 1.29\%&0&0&0\\
\noalign{\smallskip}\hline\noalign{\smallskip}

\multirow{5}{*}{0.02}&\multirow{5}{*}{0.06}&\multirow{5}{*}{0.03}&1&3.02\%&0&0&0\\
&&&2&3.04\%&0&0&0\\
&&&5&3.09\%&0&0&0\\
&&&10&3.10\%&1&0&0\\
&&&20&2.99\%&7&0&0\\
\noalign{\smallskip}\hline\noalign{\smallskip}

\multirow{5}{*}{0.02}&\multirow{5}{*}{0.06}&\multirow{5}{*}{0.06}&1&6.00\%&9&0&0\\
&&&2&5.99\%&7&0&0\\
&&&5&5.92\%&2&0&0\\
&&&10&5.67\%&20&0&0\\
&&&20&5.01\%&44&1&0\\
\noalign{\smallskip}\hline\noalign{\smallskip}

\multirow{5}{*}{0.02}&\multirow{5}{*}{0.12}&\multirow{5}{*}{0.01}&1&1.02\%&0&0&0\\
&&&2&1.05\%&0&0&0\\
&&&5&1.12\%&0&0&0\\
&&&10&1.22\%&0&0&0\\
&&&20&1.33\%&2&0&0\\
\noalign{\smallskip}\hline\noalign{\smallskip}

\multirow{5}{*}{0.02}&\multirow{5}{*}{0.12}&\multirow{5}{*}{0.03}&1&3.04\%&5&0&0\\
&&&2&3.08\%&5&0&0\\
&&&5&3.15\%&5&0&0\\
&&&10&3.16\%&3&0&0\\
&&&20&2.95\%&25&0&0\\
\noalign{\smallskip}\hline\noalign{\smallskip}

\multirow{5}{*}{0.02}&\multirow{5}{*}{0.12}&\multirow{5}{*}{0.06}&1&6.03\%&36&0&0\\
&&&2&6.05\%&32&0&0\\
&&&5&6.00\%&14&0&0\\
&&&10&5.68\%&35&3&0\\
&&&20&4.79\%&122&2&0\\
\noalign{\smallskip}\hline\noalign{\smallskip}

\multirow{5}{*}{0.1}&\multirow{5}{*}{0.06}&\multirow{5}{*}{0.01}&1&1.10\%&0&0&0\\
&&&2&1.19\%&0&0&0\\
&&&5&1.43\%&0&0&0\\
&&&10&1.68\%&1&0&0\\
&&&20&1.88\%&3&0&0\\
\noalign{\smallskip}\hline\noalign{\smallskip}

\multirow{5}{*}{0.1}&\multirow{5}{*}{0.06}&\multirow{5}{*}{0.03}&1&3.11\%&1&0&0\\
&&&2&3.18\%&2&0&0\\
&&&5&3.27\%&2&0&0\\
&&&10&3.15\%&4&0&0\\
&&&20&2.81\%&9&0&0\\
\noalign{\smallskip}\hline\noalign{\smallskip}

\multirow{5}{*}{0.1}&\multirow{5}{*}{0.06}&\multirow{5}{*}{0.06}&1&6.00\%&14&0&0\\
&&&2&5.94\%&13&0&0\\
&&&5&5.52\%&0&0&0\\
&&&10&4.72\%&17&1&0\\
&&&20&3.73\%&19&0&0\\
\noalign{\smallskip}\hline\noalign{\smallskip}

\multirow{5}{*}{0.1}&\multirow{5}{*}{0.12}&\multirow{5}{*}{0.01}&1&1.12\%&0&0&0\\
&&&2&1.25\%&1&0&0\\
&&&5&1.54\%&3&0&0\\
&&&10&1.74\%&2&1&0\\
&&&20&1.78\%&11&2&1\\
\noalign{\smallskip}\hline\noalign{\smallskip}

\multirow{5}{*}{0.1}&\multirow{5}{*}{0.12}&\multirow{5}{*}{0.03}&1&3.19\%&10&0&0\\
&&&2&3.33\%&14&0&0\\
&&&5&3.43\%&12&0&0\\
&&&10&3.16\%&7&3&0\\
&&&20&2.63\%&20&2&0\\
\noalign{\smallskip}\hline\noalign{\smallskip}

\multirow{5}{*}{0.1}&\multirow{5}{*}{0.12}&\multirow{5}{*}{0.06}&1&6.15\%&48&0&0\\
&&&2&6.16\%&49&0&0\\
&&&5&5.67\%&21&4&0\\
&&&10&4.64\%&23&7&1\\
&&&20&3.47\%&37&3&0\\
\noalign{\smallskip}\hline\noalign{\smallskip}

\end{tabular}

\end{table}



\begin{figure}
\resizebox{\hsize}{!}{\includegraphics*{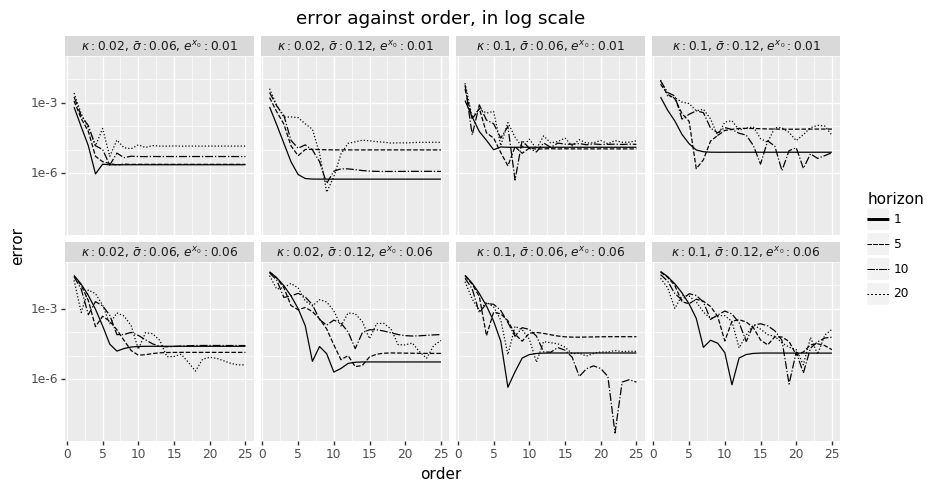}}
\caption{Absolute approximation error of zero-coupon bond price yield against Monte Carlo simulation. Different lines represent tenors. For different parameters the unconditional mean $\bar{\mu}$ is kept constant at $0.03$.}
\label{fig:orders}
\end{figure}

\section{Applications to credit derivatives}\label{sec:appshotcredit}

We follow the setup of the generalized Markovian model of credit rating migrations introduced by \citet{lando1998cox}. This model assumes that companies migrate independently within a set of $m$ ratings, e.g.~$\{\mathrm{AAA},\mathrm{AA},\ldots,\mathrm{CCC},\mathrm{D}\}$, where AAA is the highest quality rating and D represents default, or $\{\mathrm{IG},\mathrm{HY},\mathrm{D}\}$ for investment grade and high-yield bonds. A company's rating at time $R_t$ follows a continuous-time Markov chain with the ratings as states, and with $m\times m$ generator matrix $Q(Y_t)$ that depends on a latent driving process $Y_t$ of state variables.

Consider the $m\times m$ rating migration probability matrix conditional on the full history of state variables, i.e.
\begin{align}
    P_{ij}^Y(t)\coloneqq\mathbb{P}\left[R_t=j\middle|R_0=i,\mathcal{F}^Y_t\right],
\end{align}
where $\mathcal{F}^Y_t\coloneqq\sigma(Y_s,0\le s\le t)$ is the natural filtration of the stochastic process $Y_t$. Then $P^Y\!(t)$ follows the Kolmogorov forward equation
\begin{align}
    \partial_tP^Y\!(t)&=P^Y\!(t)Q(Y_t),&P^Y\!(0)&=I_m.
\end{align}
In order to derive credit spreads and rating migration probabilities, the interest is in the \emph{rating migration matrix}
\begin{align}
    P_{ij}(t,y)\coloneqq\mathbb{P}_y[R_t=j|R_0=i]
    =\mathbb{E}_y\!\left[P^Y_{ij}(t)\right].
\end{align}

As shown by \citet{lando1998cox} in the specific case that $Q(Y_t)$ commute and are diagonalizable, we can solve the Kolmogorov forward equation and write
\begin{align}
    P^Y\!(t)=e^{\int_0^t{Q(Y_s)\ud s}}
    =B\,e^{\int_0^t{D(Y_s)\ud s}}B^{-1}
    =B\,\mathrm{diag}_i\!\left(e^{\int_0^t{D_{ii}(Y_s)\ud s}}\right)B^{-1},
\end{align}
with diagonalization $Q(y)=BD(y)B^{-1}$ where $D(y)$ is a diagonal matrix of (non-positive) eigenvalues. This strategy uses the fact that commuting diagonalizable matrices are simultaneously diagonalizable, i.e.~share the matrix $B$. Taking the expectation results in a set of bond price-like formulas that can be solved analytically in certain cases, i.e.
\begin{align}\label{eqn:anl1}
    P(t,y)
    =\mathbb{E}_y\!\left[e^{\int_0^t{Q(Y_s)\ud s}}\right]
    =B\,\mathrm{diag}_i\!\left(\mathbb{E}_y\!\left[e^{\int_0^t{D_{ii}(Y_s)\ud s}}\right]\right)B^{-1}.
\end{align}

In this section, we will \emph{not} assume such commuting property of the generators, and use the proposed approximation strategy to calculate the rating migration matrix. To this end, we define the basis vector-valued process $Z_t=e_{R_t}$ with state space of $m$-dimensional basis vectors $E=\{e_1,\ldots,e_m,\}$. We assume that $Y_t$ follows an $n$-dimensional time-homogenous It\^{o} diffusion with state space $D'$. The SDE for the joint process $X\coloneqq(Y,Z)$ is
\begin{align}\label{eqn:sde_credit}
\ud Y_t &= \mu(Y_t)\ud t+\sigma(Y_t)\ud W_t, \\
\ud Z_t &= \sum_{i=1}^m{Z_{i,t-}\sum_{j\neq i}{(e_j-e_i)}}\ud N_t^{ij},\nonumber
\end{align}
where $N_t^{ij}$ are Poisson processes with intensity $\mathbb{E}[\ud N_t^{ij}|\mathcal{F}_t^Y]=Q_{ij}(Y_t)\ud t$. Intuitively if the Markov chain $R_t$ is in state $i$ then $Z_{i,t-}=1$ and it migrates to state $j\neq i$ with intensity $Q_{ij}(Y_t)$. A jump to state $j$ modifies $Z_t$ by subtracting the current state $e_i$ and adding the new state $e_j$. It follows from basic manipulation that
\begin{align}
\ud Z_t &= Q(Y_t)^\top Z_{t-}\ud t+\ud M_t,
\end{align}
with $M_t$ a martingale.\footnote{To see this use $Q(Y_t)1=0$ to get
\begin{align*}
\mathbb{E}[\ud Z_t|\mathcal{F}_t]
&=\sum_{i=1}^m{Z_{i,t-}\sum_{j\neq i}{(e_j-e_i)}}\mathbb{E}[\ud N_t^{ij}|\mathcal{F}_t^Y]
=\sum_{i=1}^m{Z_{i,t-}\sum_{j=1}^m{(e_j-e_i)}}Q_{ij}(Y_t)\ud t
\\
&=\sum_{i=1}^m{Z_{i,t-}\sum_{j=1}^m{e_j}}Q_{ij}(Y_t)\ud t
-\sum_{i=1}^m{Z_{i,t-}\sum_{j=1}^m{e_i}}Q_{ij}(Y_t)\ud t\\
&=\sum_{j=1}^m{\left(\sum_{i=1}^m{Z_{i,t-}Q_{ij}(Y_t)}\right)e_j}\ud t
-\sum_{i=1}^m{Z_{i,t-}e_i(Q(Y_t)1)_i}\ud t\\
&=\sum_{j=1}^m{(Z_{t-}^\top Q(Y_t)e_j)e_j}\ud t
=\sum_{j=1}^m{(Q(Y_t)^\top Z_{t-})_je_j}\ud t
=Q(Y_t)^\top Z_{t-}\ud t.
\end{align*}}
In this setting, we can express the rating migration matrix as an expectation that conforms our approximation approach,
\begin{align*}
P_{ij}(t,y)&=\mathbb{P}_y[Z_t=e_j|Z_0=e_i]=\mathbb{E}_{y,e_i}[\langle e_j,Z_t\rangle]=S_tf(y,e_i)
\end{align*}
where $f(y,z)=z_j$. In order to apply the approximation, we need the generator of the process $X$. If $Y$ has generator $\mathcal{A}^y$, then following standard arguments, the generator of the process $X$ is
\begin{align}
\mathcal{A}f(y,z)&=\mathcal{A}^yf(y,z)+\mathcal{A}^zf(y,z),&\mathcal{A}^zf(y,z)&\coloneqq z^\top Q(y)\begin{bmatrix}f(y,e_1)\\\vdots\\f(y,e_m)\end{bmatrix}.
\end{align}
It is easy to see that with $b(y)$ an appropriate basis for $Y$, $b(x)\coloneqq b(y)\otimes z$ is appropriate for $X$. In what follows we will assume that $b(y)=(1,y,y^2,\ldots)$ such that $b(x)=(z,y\otimes z,y^2\otimes z,\ldots)$. With $\ell$ being the number of powers of $y$ included, for $n=1$ this basis has dimension $\ell m$, for $n=2$ the dimension is $\ell(\ell+1)m/2$, and in general for $n>1$ the dimension is $\sum_{i=0}^{\ell-1}\binom{n+i-1}{n-1}m$.

The steps in approximating $P(t,y)$ are as follows. First we derive the matrix equivalent of the generator, $A$, by applying the generator $\mathcal{A}$ to the basis $b(x)$, and construct the finite approximation $A_k$. Then $P_{ij}(t,y)=S_tf(y,e_i)$ with $f(y,z)=z_j=\langle(e_j,0,\ldots),b(y,z)\rangle$. Therefore,
\begin{align}
    P_{ij}(t,y)&\approx\left\langle e^{tA_k}(e_j,0)^\top,b^k(y,e_i)\right\rangle,\\
    P(t,y)&\approx(b^{k/m}(y)^\top\otimes I_m)e^{tA_k}(I_m,O)^\top.
\end{align}

Due to the additive structure of the generator, we can come up with a few simplifying results to aid the construction of $A$ and $A_k$. To do so, we split $A=A^y+A^z$ in accordance with the decomposition $\mathcal{A}=\mathcal{A}^y+\mathcal{A}^z$. The following proposition establishes that we can compute $A^y$ in isolation, i.e.~without knowing the generator matrix $Q(y)$ or its dimension $m$. 

\begin{proposition}\label{prp:Ay}
Consider the sequential process $Y$ with infinitesimal generator $\mathcal{A}^y$ in accordance with the specification above. Derive $A^1$ by applying $\mathcal{A}^y$ to the basis $b(y)$ assuming $Q(y)=0$ and $m=1$. Then for general $Q(y)$ it holds that
\begin{align}
A^y = A^1 \otimes I_m.
\end{align}
\end{proposition} 

\subsection{Migrations driven by multivariate CIR processes}

Consider the process $X=(Y,Z)$ on domain $D=\mathbb{R}^n_+\times E$ where $E=\{e_1,\ldots,e_m\}$ such that $d=n+m$ through its SDE,
\begin{align}\label{eqn:sde_credit}
\ud Y_t &= K( \mu-Y_t)\mathrm{d} t +\mathrm{diag}_i \left( \sigma_i\sqrt{Y_{i,t}} \right)\mathrm{d} W_t, \\
\ud Z_t &= Q(Y_t)^\top Z_t\mathrm{d} t+\mathrm{d} M_t,\nonumber
\end{align}
where $Y_t$ follows a multivariate CIR process with $n\times n$ mean reversion speed $K$ to means $\mu$. In case $K$ is diagonal, $Y_{i,t}$ are $n$ independent CIR processes. We choose the generator matrix $Q(y)\coloneqq\sum_{i=1}^ny_iQ_i$, where $Q_i$ are generator matrices of continuous-time Markov chains. Since $Y_t$ is non-negative $Q(Y_t)$ is a well-defined generator matrix.

In this specific case, the generator of the process $X$ is
\begin{align*}
\mathcal{A}f(y,z)
&=\sum_{i=1}^n\left(\frac12\sigma_i^2 y_i\frac{\partial^2 f(y,z)}{\partial y_i^2}
+e_i^\top K(\mu-y)\frac{\partial f(y,z)}{\partial y_i}\right)
\\&\phantom{=}+\sum_{i=1}^n{y_iz^\top Q_i\begin{bmatrix}f(y,e_1)\\\vdots\\f(y,e_m)\end{bmatrix}}
\end{align*}
Applying $\mathcal{A}^y$ to elements of $b^k(y)$ always returns order of at most $x^{k-1}$, hence $A_y^1$ is upper triangular, and so is $A^y$. $A^z$ is not upper triangular, and thus the projection method matters here. We choose Taylor approximation around the means $y_0=\mu$ in all cases.


\subsubsection{Migrations driven by a univariate CIR process}

\citet{arvanitis1999building} apply the model above to the univariate case, i.e.~$n=1$. With diagonalization $Q_1=BDB^{-1}$, we use (\ref{eqn:anl1}) to obtain
\begin{align}\label{eqn:anl2}
    P(t,y)=B\,\mathrm{diag}_i\!\left(\mathbb{E}_y\!\left[e^{\int_0^t{D_{ii}Y_s\ud s}}\right]\right)B^{-1}.
\end{align}
Since $Y$ follows a CIR process and $D_{ii}$ are the non-positive eigenvalues of $Q_1$, $-Y_sD_{ii}$ is either 0 or follows a CIR process. Thus each diagonal element is either 1 or a CIR bond price, which has an analytical solution. Then $P(t,y)$ has an analytical solution that we can use as a benchmark for our approximation.

The matrix $A=A^y+A^z=A^1\otimes I_m+A^z$ has the following components, as also found in \citet{cuchiero2012polynomial},
\begin{align*}
A_k^1=\begin{bmatrix}
0&K\mu&0&0&\cdots&0\\
0&-K&2K\mu+\sigma^2&0&\cdots&0\\
0&0&-2K&3K\mu+3\sigma^2&\cdots&0\\
\vdots&\ddots&\ddots&\ddots&\ddots&\vdots\\
0&\cdots&0&0&-(k-2)K&(k-2)K\mu+\tfrac12(k-2)(k-3)\sigma^2\\
0&\cdots&0&0&0&-(k-1)K
\end{bmatrix}
\end{align*}
The structure of $A_z$ requires the same Taylor approximation as in the CIR bond price case,
\begin{align*}
A^z_k =
\begin{bmatrix}
O&\cdots&O&\bar{p}_1^{k/m}(\mu)Q_1\\
Q_1&\cdots&O&\bar{p}_2^{k/m}(\mu)Q_1\\
\vdots&\ddots&\vdots&\vdots\\
O&\cdots&Q_1&\bar{p}_{k/m}^{k/m}(\mu)Q_1
\end{bmatrix}
=
\begin{bmatrix}
\begin{matrix}0^\top\\I\end{matrix}\vline&\bar{p}^{k/m}(\mu)
\end{bmatrix}\otimes Q_1,
\end{align*}
where the $k/m$-order Taylor approximation of $y^{k/m}$ around $y_0=\mu$ defines the last columns, with Taylor coefficients $\bar{p}_{i+1}^{k/m}(y_0)=-{k/m \choose i}(-y_0)^{k/m-i}$.


To test accuracy of proposed approximation, we choose the same Markov chain generator matrix $Q_1$ as \citet[Example 1]{jarrow1997markov},
\begin{align*}
Q_1=\begin{bmatrix}
-0.11 & 0.1   & 0.01\\
0.05  & -0.15 & 0.1 \\
0     &0      & 0
\end{bmatrix}.
\end{align*}
Figure \ref{fig:orders1d} plots the mean absolute approximation error across all matrix entries against the order of the approximation $\ell=k/m$, for various parameters of the CIR process and maturities.
\begin{figure}
\resizebox{\hsize}{!}{\includegraphics*{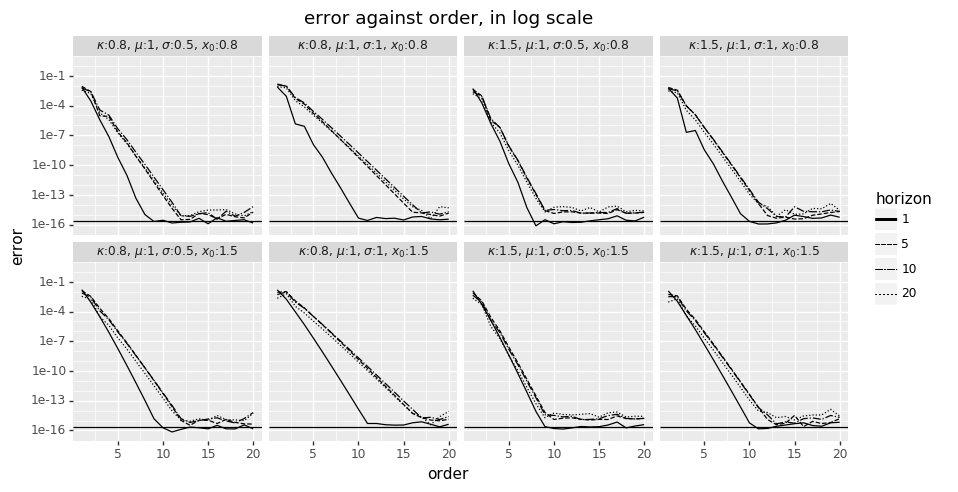}}
\caption{One dimension, migration probability element-wise mean absolute error against different orders, compare with analytical solution, in log scale. Different panels stand for different parameters. Different lines stand for different time horizon. Horizontal lines stand for machine precision.}
\label{fig:orders1d}
\end{figure}
We can see the error is decreasing exponentially as the order increases. Convergence appears to be faster for shorter horizons. All errors converge to one or two orders above machine precision.

\subsubsection{Migrations driven by a bivariate CIR process, commuting case}

\citet{hurd2007affine} apply the same credit model to the bivariate case i.e.~$n=2$. In order to ensure tractability, they specify $Q_2$ to reflect an additional default migration that is the same for all ratings. This second generator matrix can also be interpreted as a liquidity premium. The specific structure of this second matrix ensures that $Q_1$ and $Q_2$ commute. With the additional assumption that $K$ is diagonal, i.e.~the two driving CIR processes are independent, we can write
\begin{align*}
    P(t,y)=\mathbb{E}_y\left[e^{\int_0^t{(Y_{1,s}Q_1+Y_{2,s}Q_2)\ud s}}\right]
    =\mathbb{E}_{y_1}\left[e^{\int_0^t{Y_{1,s}Q_1\ud s}}\right]\mathbb{E}_{y_2}\left[e^{\int_0^t{Y_{2,s}Q_2\ud s}}\right],
\end{align*}
and apply the univariate pricing strategies.

We keep $Q_1$ the same as in the previous example, and follow \citet[Section 7]{hurd2007affine} to define $Q_2$, i.e.
\begin{align*}
Q_1&=\begin{bmatrix}
-0.11 & 0.1   & 0.01\\
0.05  & -0.15 & 0.1 \\
0     &0      & 0
\end{bmatrix},&
Q_2&=\begin{bmatrix}
-0.01 & 0     & 0.01 \\
0     & -0.01 & 0.01 \\
0     &0      & 0
\end{bmatrix}.
\end{align*}
For all other parameters, $K_{big}=\diag(1.5,1.5)$, $K_{small}=\diag(0.8,0.8)$, $\sigma_{big}=(1.0,1.0)$, $\sigma_{small}=(0.5,0.5)$, $x_{small}=(0.8,0.8)^\top$, and $x_{big}=(1.2,1.2)^\top$. The parenthesis lists the parameters of the two independent CIR processes. Figure \ref{fig:orders2d} shows the same pattern of approximation quality as in the univariate case.
\begin{figure}
\resizebox{\hsize}{!}{\includegraphics*{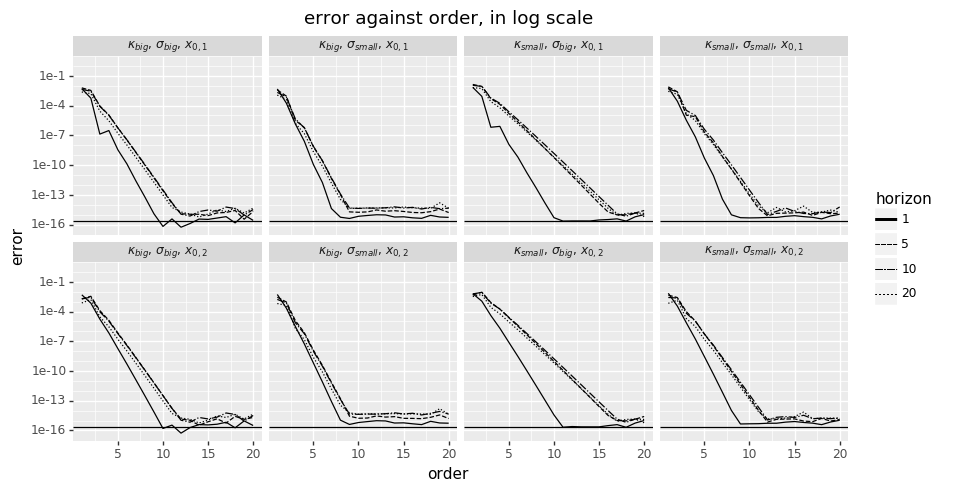}}
\caption{Two dimension when analytical solution exists, migration probability element-wise mean absolute error against different orders, compare with analytical solution, in log scale. Different panels stand for different parameters. Different lines stand for different time horizon. Horizontal lines stand for machine precision.}
\label{fig:orders2d}
\end{figure}

\subsubsection{Migrations driven by a bivariate CIR process, non-commuting case}

We consider a case for which no analytical solution exists, namely when both the independence and commutativity requirement fail. Let $Q_1$ and $Q_2$ be upper and lower-triangular respectively. This means that $Y_1$ and $Y_2$ represent the scaling processes that accelerate the speed of upgrades and downgrades separately. This model can capture an important stylized fact, that with the business cycle upgrades tend to slow down when downgrades speed up, and vice versa. The details of this model are discussed in an accompanying empirical paper \citep{ba2020integrated}. 

Since no analytical solution exists, we benchmark the approximation against a Monte Carlo simulation. We choose $Q_1$ and $Q_2$ to be upper and lower triangular decomposition of matrix in \citet[Example 1]{jarrow1997markov},
\begin{align*}
Q_1&=\begin{bmatrix}
-0.11 & 0.1   & 0.01\\
0     & -0.1  & 0.1 \\
0     & 0     & 0   \\
\end{bmatrix}, &
Q_2&=\begin{bmatrix}
0     & 0     & 0\\
0.05  & -0.05 & 0\\
0     &0      & 0
\end{bmatrix}.
\end{align*}
These matrices do not commute.\footnote{Using the \citet[Conjecture 1.2]{bottcher2005big} we get the Frobenius norm inequality $\|Q_1Q_2-Q_2Q_1\|_F\leq\sqrt{2}\|Q_1\|_F\|Q_2\|_F$. This inequality gives rise to a measure of non-commutativity for non-trivial matrices with values in $[0,1]$,
\begin{align*}
\frac{\|Q_1Q_2-Q_2Q_1\|_F}{\sqrt{2}\|Q_1\|_F\|Q_2\|_F}=0.48.
\end{align*}
This value allows us to conclude that these matrices are strongly non-commuting.}

All other parameters are  $K_{big}=\begin{bmatrix}1.5&0.4\\0.4&1.5\end{bmatrix}$, $K_{small}=\begin{bmatrix}0.8&0.2\\0.2&0.8\end{bmatrix}$, $\sigma_{big}=(1.0,1.0)$, $\sigma_{small}=(0.5,0.5)$,  $x_{small}=(0.8,0.8)^\top$, and $x_{big}=(1.2,1.2)^\top$. Note that the matrix $K$ is no longer diagonal, and can induce correlation between the upgrade and downgrade speed processes.

Figure \ref{fig:orders2db} shows the mean absolute approximation error against the approximation order $\ell=k/m$.
\begin{figure}
\resizebox{\hsize}{!}{\includegraphics*{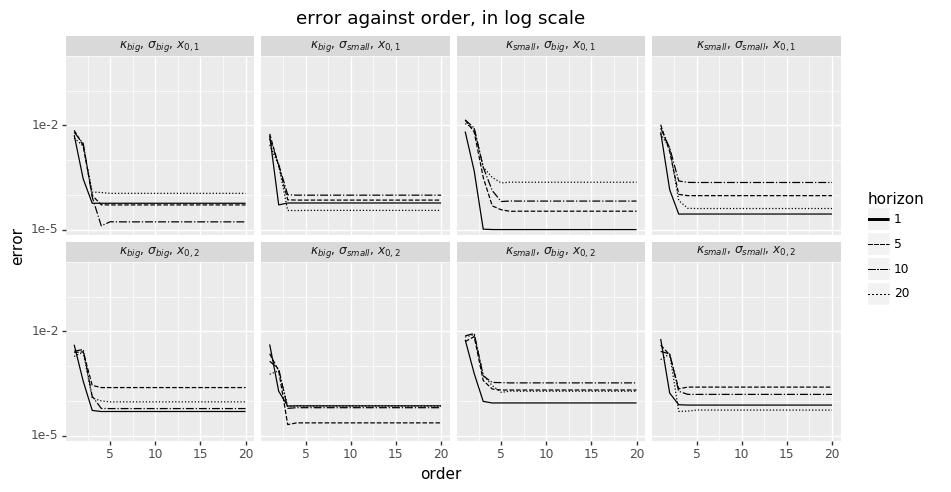}}
\caption{Two dimension when analytical solution does not exist, migration probability element-wise mean absolute error against different orders, compare with MC simulation, in log scale. Different panels stand for different parameters. Different lines stand for different time horizon.}
\label{fig:orders2db}
\end{figure}
It is worth noting that the errors converge to bounds within the margin of error expected from the Monte Carlo simulation. Overall, we see the same pattern of exponential decline in the error, although as expected a Monte Carlo-induced lower bound is hit sooner than when comparing against an analytical benchmark.


\appendix

\section{Appendix}\label{appendix}

The proof of Theorem \ref{thm:main} relies on two Lemmas that are presented below.

The first lemma is an inversion result \`a la Phragm\'en-Doetsch for Laplace transforms, see \citet{arendt1987vector} and Phragm\'en's approximation result for semigroups of operators, see~\citet{neubrander1987relation}. The general setting is that $(S_t)$ is a strongly continuous semigroup acting on a Banach space $\mathcal{P}$ with generator $\mathcal{A}$ having domain $\mathcal{D}(\mathcal{A})$. Let $R(\lambda,\mathcal{A})$ denote the resolvent, which exists for all sufficiently large $\lambda$.
We will use that there exist $C\geq 1$ and $w\geq 0$ be such that the operator norm 
\begin{equation}\label{eq:cw}
\|S_t\|\leq Ce^{wt}, 
\end{equation}
see \citet[7.14]{bobrowski2005functional}.

\begin{lemma}\label{lemma:pd}
Let $t>0$ and
\begin{equation}\label{eq:defs}
S_t(\lambda,\mathcal{A},f):=\lambda\sum_{n=1}^\infty (-1)^{n-1}\frac{1}{(n-1)!}e^{n\lambda t}R(n\lambda,\mathcal{A})f. 
\end{equation}
Then for some $w\geq 0$ and $C\geq 1$ it holds that for $\lambda>2w$ and all $f\in \mathcal{D}(\mathcal{A})$,
\[
\|S_tf- S_t(\lambda,\mathcal{A},f)\| \leq  C\frac{e^{wt}}{\lambda}\|f\|.
\]
\end{lemma}

\begin{proof}
Denote 
\[
\tilde{S}_t(\lambda,\mathcal{A},f)\coloneqq\sum_{n=1}^\infty(-1)^{n-1}\frac{1}{n!}e^{n\lambda t}R(n\lambda,\mathcal{A})f.
\]
A first result, adapted from~\citet{neubrander1987relation} and from Theorem~2.3.2 (Phragm\'en-Doetsch Inversion) in \citet{arendt1987vector} for Laplace transforms of functions, is that for $f\in\mathcal{P}$ and $\lambda>w$ one has 
\[
\left\|\int_0^t S_uf\ud u- 
\tilde{S}_t(\lambda,\mathcal{A},f)\right\|
\leq C\frac{e^{wt}}{\lambda-w}\|f\|,
\]
where $C$ and $w$ are as in \eqref{eq:cw}.
Following the line of thoughts as in  \citep{neubrander1987relation}, we have for $f\in\mathcal{D}(\mathcal{A})$ the equality $R(n\lambda,\mathcal{A})\mathcal{A}f=n\lambda R(n\lambda,\mathcal{A})-f$ and hence
\begin{align*}
\tilde{S}_t(\lambda,\mathcal{A},\mathcal{A}f) & = S_t(\lambda,\mathcal{A},f)- \sum_{n=1}^\infty(-1)^{n-1}\frac{1}{n!}e^{n\lambda t} f \\
 & = S_t(\lambda,\mathcal{A},f)+ \left(\exp{\{-e^{\lambda t}\}}-1\right)f.
\end{align*}
Using $S_tf=f+\int_0^tS_u\mathcal{A}f$ for $f\in\mathcal{D}(\mathcal{A})$, we develop for $\lambda>2w$
\begin{align*}
\|S_tf- S_t(\lambda,\mathcal{A},f)\| & \leq  \|f\|+ \left\|\int_0^t{S_u\mathcal{A}\ud u} - \tilde{S}_t(\lambda,\mathcal{A},f)\right\| +\|\tilde{S}_t(\lambda,\mathcal{A},f)-S_t(\lambda,\mathcal{A},f) \| \\
& \leq  C\frac{e^{wt}}{\lambda-w}\|f\|+ \exp{\{-e^{\lambda t}\}}\|f\| \\
& \leq \left(2C\frac{e^{wt}}{\lambda}+ \exp{\{-e^{\lambda t}\}}\right)\|f\| \\ 
& \leq C'\frac{e^{wt}}{\lambda}\|f\|,
\end{align*}
for some constant $C'$. 
\end{proof}
If the space $\mathcal{P}$ is the Hilbert space as in Sections~\ref{sec:notation} and \ref{sec:theory}, we can write the counterpart of Lemma~\ref{lemma:pd} for sequences in $\mathcal{H}$. If the $S_t$ form the strongly continuous transition semigroup on $\mathcal{P}$ of a Feller process, then they are all expectations and we can take $C=1$ and $w=0$ in \eqref{eq:cw}. The same is true for the induced semigroup of operators $\bar{S}_t$ on $\mathcal{H}$. This leads to the following variation on Lemma~\ref{lemma:pd}.

\begin{corollary}\label{corollary:pd}
Let $t>0$ and
\begin{equation}\label{eq:defs2}
S_t(\lambda,A,\bar{f}):=\lambda\sum_{n=1}^\infty (-1)^{n-1}\frac{1}{(n-1)!}e^{n\lambda t}R(n\lambda,A)\bar{f}. 
\end{equation}
Then for some $w\geq 0$ and $C\geq 1$ it holds that for $\lambda\geq 2w$ and all $\bar{f}\in \mathcal{D}(A)$,
\[
\|\bar{S}_t\bar{f}- S_t(\lambda,A,\bar{f})\| \leq  C\frac{e^{wt}}{\lambda}\|\bar{f}\|.
\]
Moreover, if the $S_t$ form the strongly continuous transition semigroup of a Feller process, then 
\[
\|\bar{S}_t\bar{f}- S_t(\lambda,A,\bar{f})\| \leq  \frac{C}{\lambda}\|\bar{f}\|.
\]
\end{corollary}
In the next lemma we specialize to the situation where the semigroup acts on elements of a Hilbert space. So we assume that $(S_t)$ be a Feller semigroup defined on a Hilbert space $\mathcal{H}$ with generator $A$. Let $P_k$ be projections of $\mathcal{H}$ on to $\mathcal{H}_k$ with norm $\|P_k\|=1$, typically orthogonal projections. Let $A_k$ be as in \eqref{def:ak}. Let $R(\lambda,A)$ and $R(\lambda,A_k)$ be the corresponding resolvents. For $f^k\in\mathcal{D}(A_k)$ we consider $R(\lambda,A_k)\bar{f}^k\in\mathcal{H}_k$ as an element of $\mathcal{H}$.

\begin{lemma}\label{lemma:sks}
In the setting just described, assume that the operators $A_k:\mathcal{H}_k\to \mathcal{H}_k$ are such that with $\bar{f}^k=P_k\bar{f}$, $\bar{f}\in\mathcal{D}(A)$ and $\bar{f}^k\in\mathcal{D}(A_k)$ satisfy $\lim_{k\to\infty} R(n\lambda,A_k)\bar{f}^k=R(n\lambda,A)\bar{f}$, for all $n\geq1$ and where the limit is taken in $\mathcal{H}$. 
Then $S_t(\lambda,A_k,\bar{f}^k)\to S_t(\lambda,A,\bar{f})$ for $k\to\infty$.\end{lemma}

\begin{proof}
Recall from \eqref{eq:cw} that there exist $C,w>0$ such that $\|S_t\|\leq Ce^{wt}$, from which it follows that $\|R(\lambda,A)\|\leq \frac{C}{\lambda -w}$, which is at most equal to $\frac{2C}{\lambda}$ for all $\lambda \geq2w$. It follows that then $\| R(n\lambda,A)\bar{f}\|\leq \frac{2C}{n\lambda}\|\bar{f}\|$ whenever $n\lambda\geq 2w$. Since the $A_k$ are obtained from $A$ by the projections, we have also have $\| R(n\lambda,A_k)\bar{f}^k\|\leq \frac{2C}{n\lambda}\|\bar{f}^k\|\leq \frac{2C}{n\lambda}\|\bar{f}\|$, since $\|\bar{f}^k\|\leq \|\bar{f}\|$.

Let $\{S^k_t\}_{t\geq0}$ be the semigroup generated by $A^k$ acting on sequences $\bar{f}^k= P_k\bar{f}$. One easily verifies that $P_kA_k=A_kP_k$ and then $(P_kAP_k)^j=P_kA^jP_k$ which then leads to $\exp{\{tA_k\}}P_kf=S^k_tP_kf=P_kS^k_tP_kf=P_kS_tP_kf$. By the representation of resolvents as Laplace transforms, one obtains $R(\lambda,A_k)P_kf=P_kR(\lambda,A)P_kf$ and then $\| R(n\lambda,A_k)\bar{f}^k\|\leq \frac{2C}{n\lambda}\|\bar{f}^k\|\leq \frac{2C}{n\lambda}\|\bar{f}\|$, since $\|\bar{f}^k\|\leq \|\bar{f}\|$.

Consider the norm of the summands in $S_t(\lambda,A_k,\bar{f}_k)$. For each $n$ this norm is at most equal to
\[
\lambda\frac{e^{n\lambda t}}{(n-1)!}\|R(n\lambda,A_k)\bar{f}^k\|\leq \lambda\frac{e^{n\lambda t}}{(n-1)!}\frac{2C}{n\lambda}\|\bar{f}\|=\frac{e^{n\lambda t}}{n!}\frac{2C}{\lambda}\|\bar{f}\|,
\]
which has a finite sum over $n\geq 1$. Hence, considering the infinite sum 
\begin{equation}\label{eq:sum2}
S_t(\lambda,A_k,\bar{f}_k)=\lambda\sum_{n=1}^\infty (-1)^{n-1}\frac{1}{(n-1)!}e^{n\lambda t}R(n\lambda,A_k)\bar{f}^k 
\end{equation}
as a Bochner integral, we can apply dominated convergence for Bochner integrals (see~\cite[Proposition~1.2.5]{hytonen2016analysis}), to \eqref{eq:sum2} to arrive at the convergence
\[
\lambda\sum_{n=1}^\infty(-1)^{n-1}\frac{e^{n\lambda t}}{(n-1)!}R(n\lambda,A_k)\bar{f}^k \to \lambda\sum_{n=1}^\infty(-1)^{n-1}\frac{e^{n\lambda t}}{(n-1)!}R(n\lambda,A)\bar{f},
\]
which was our aim.
\end{proof}

\paragraph{Acknowledgements}
We express our gratitude to Richard Martin, Andrew Ang and Yury Krongauz for valuable input on our paper, including the suggested application to Gram-Charlier series.

%
%

\bibliographystyle{plainnat}
\bibliography{literature}   

\begin{thebibliography}{21}
\providecommand{\natexlab}[1]{#1}
\providecommand{\url}[1]{\texttt{#1}}
\expandafter\ifx\csname urlstyle\endcsname\relax
  \providecommand{\doi}[1]{doi: #1}\else
  \providecommand{\doi}{doi: \begingroup \urlstyle{rm}\Url}\fi

\bibitem[A{\"\i}t-Sahalia(2002)]{ait2002maximum}
Yacine A{\"\i}t-Sahalia.
\newblock Maximum likelihood estimation of discretely sampled diffusions: a
  closed-form approximation approach.
\newblock \emph{Econometrica}, 70\penalty0 (1):\penalty0 223--262, 2002.

\bibitem[Al-Mohy and Higham(2011)]{al2011computing}
Awad~H Al-Mohy and Nicholas~J Higham.
\newblock Computing the action of the matrix exponential, with an application
  to exponential integrators.
\newblock \emph{SIAM journal on scientific computing}, 33\penalty0
  (2):\penalty0 488--511, 2011.

\bibitem[Arendt(1987)]{arendt1987vector}
Wolfgang Arendt.
\newblock Vector-valued {L}aplace transforms and {C}auchy problems.
\newblock \emph{Israel Journal of Mathematics}, 59\penalty0 (3):\penalty0
  327--352, 1987.

\bibitem[Arvanitis et~al.(1999)Arvanitis, Gregory, and
  Laurent]{arvanitis1999building}
Angelo Arvanitis, Jonathan Gregory, and Jean-Paul Laurent.
\newblock Building models for credit spreads.
\newblock \emph{The Journal of Derivatives}, 6\penalty0 (3):\penalty0 27--43,
  1999.

\bibitem[Ba et~al.(2020)Ba, van Beek, and Zhao]{ba2020integrated}
Makhtar Ba, Misha van Beek, and Chenyu Zhao.
\newblock An integrated credit model.
\newblock Working paper, BlackRock, New York, 2020.

\bibitem[Black and Karasinski(1991)]{black1991bond}
Fischer Black and Piotr Karasinski.
\newblock Bond and option pricing when short rates are lognormal.
\newblock \emph{Financial Analysts Journal}, 47\penalty0 (4):\penalty0 52--59,
  1991.

\bibitem[Bobrowski(2005)]{bobrowski2005functional}
Adam Bobrowski.
\newblock \emph{Functional analysis for probability and stochastic processes:
  an introduction}.
\newblock Cambridge University Press, 2005.

\bibitem[B{\"o}ttcher and Wenzel(2005)]{bottcher2005big}
Albrecht B{\"o}ttcher and David Wenzel.
\newblock How big can the commutator of two matrices be and how big is it
  typically?
\newblock \emph{Linear algebra and its applications}, 403:\penalty0 216--228,
  2005.

\bibitem[Cox et~al.(2005)Cox, Ingersoll~Jr, and Ross]{cox2005theory}
John~C Cox, Jonathan~E Ingersoll~Jr, and Stephen~A Ross.
\newblock A theory of the term structure of interest rates.
\newblock In \emph{Theory of valuation}, pages 129--164. World Scientific,
  2005.

\bibitem[Cuchiero et~al.(2012)Cuchiero, Keller-Ressel, and
  Teichmann]{cuchiero2012polynomial}
Christa Cuchiero, Martin Keller-Ressel, and Josef Teichmann.
\newblock Polynomial processes and their applications to mathematical finance.
\newblock \emph{Finance and Stochastics}, 16\penalty0 (4):\penalty0 711--740,
  2012.

\bibitem[Duffie et~al.(2003)Duffie, Filipovi{\'c}, Schachermayer,
  et~al.]{duffie2003affine}
Darrell Duffie, Damir Filipovi{\'c}, Walter Schachermayer, et~al.
\newblock Affine processes and applications in finance.
\newblock \emph{The Annals of Applied Probability}, 13\penalty0 (3):\penalty0
  984--1053, 2003.

\bibitem[Hurd and Kuznetsov(2007)]{hurd2007affine}
Tom Hurd and Alexey Kuznetsov.
\newblock Affine {M}arkov chain models of multifirm credit migration.
\newblock \emph{Journal of Credit Risk}, 3\penalty0 (1):\penalty0 3--29, 2007.

\bibitem[Hyt{\"o}nen et~al.(2016)Hyt{\"o}nen, Van~Neerven, Veraar, and
  Weis]{hytonen2016analysis}
Tuomas Hyt{\"o}nen, Jan Van~Neerven, Mark Veraar, and Lutz Weis.
\newblock \emph{Analysis in {B}anach Spaces Volume {I}: {M}artingales and
  {L}ittlewood-{P}aley Theory}, volume~12.
\newblock Springer, 2016.

\bibitem[Jarrow et~al.(1997)Jarrow, Lando, and Turnbull]{jarrow1997markov}
Robert~A Jarrow, David Lando, and Stuart~M Turnbull.
\newblock A {M}arkov model for the term structure of credit risk spreads.
\newblock \emph{The review of financial studies}, 10\penalty0 (2):\penalty0
  481--523, 1997.

\bibitem[Kulkarni and Ramesh(2008)]{kulkarni2008projection}
SH~Kulkarni and G~Ramesh.
\newblock Projection methods for inversion of unbounded operators.
\newblock \emph{Indian J. pure appl. Math}, 39:\penalty0 185--202, 2008.

\bibitem[Lando(1998)]{lando1998cox}
David Lando.
\newblock On {C}ox processes and credit risky securities.
\newblock \emph{Review of Derivatives research}, 2\penalty0 (2-3):\penalty0
  99--120, 1998.

\bibitem[Martin(2020)]{martin2020credit}
Richard~J Martin.
\newblock Credit migration: Generating generators.
\newblock \emph{arXiv preprint arXiv:2006.11146}, 2020.

\bibitem[Neubrander(1987)]{neubrander1987relation}
Frank Neubrander.
\newblock On the relation between the semigroup and its infinitesimal
  generator.
\newblock \emph{Proceedings of the American Mathematical Society}, 100\penalty0
  (1):\penalty0 104--108, 1987.

\bibitem[Popovic and Goldsman(2012)]{popovic2012easy}
Ray Popovic and David Goldsman.
\newblock Easy {G}ram-{C}harlier valuations of options.
\newblock \emph{The Journal of Derivatives}, 20\penalty0 (2):\penalty0 79--97,
  2012.

\bibitem[Tanaka et~al.(2010)Tanaka, Yamada, and
  Watanabe]{tanaka2010applications}
Keiichi Tanaka, Takeshi Yamada, and Toshiaki Watanabe.
\newblock Applications of {G}ram-{C}harlier expansion and bond moments for
  pricing of interest rates and credit risk.
\newblock \emph{Quantitative Finance}, 10\penalty0 (6):\penalty0 645--662,
  2010.

\bibitem[Zhou(2003)]{zhou2003ito}
Hao Zhou.
\newblock It{\^o} conditional moment generator and the estimation of short-rate
  processes.
\newblock \emph{Journal of Financial Econometrics}, 1\penalty0 (2):\penalty0
  250--271, 2003.

\end{thebibliography}

\end{document}